\documentclass[authoryear]{elsarticle}
\usepackage{a4}   

\usepackage[applemac]{inputenc}
\usepackage{graphicx}        
\usepackage{amsmath}
\usepackage[amsthm]{ntheorem}
\usepackage{amsfonts}

\newtheorem{definition}{Definition}[section]

\newtheorem{theorem}{Theorem}[section]

\newtheorem{corollary}{Corollary}[section]

\journal{Journal of Multivariate Analysis}

\begin{document}

\begin{frontmatter}

\title{Extremal $t$ processes: Elliptical domain of attraction and a spectral representation}

\author{T. Opitz}
\address{I3M (CC51), Universit\'e Montpellier 2 \\ Place Eug\`ene Bataillon \\
  34095 Montpellier cedex 5}
\ead{thomas.opitz@math.univ-montp2.fr}
\ead[url]{http://ens.math.univ-montp2.fr/~opitz/}
\begin{abstract}
The extremal $t$ process was proposed in the literature for modeling spatial extremes within a copula framework
based on 
the extreme value limit of elliptical $t$
distributions (Davison,
Padoan and Ribatet (2012)). A major drawback of this max-stable model was the lack of a
spectral representation  such that for instance direct simulation was infeasible.
The main contribution of this note  is to propose such a  spectral construction for the
extremal $t$ process. Interestingly, the extremal Gaussian process introduced by
Schlather (2002) appears as a special
case.
We further highlight the role of the
extremal $t$ process as the maximum attractor for processes with finite-dimensional
elliptical distributions. 
All results naturally also hold within the multivariate domain. 

\end{abstract}

\end{frontmatter}
\section{Introduction}

\cite{davison2009modelling} survey the statistical modeling of
spatial extremes and provide a global
view on available  models and their interconnections. Among these models, the extremal $t$
process represents a max-stable process that generalizes the $t$ extreme
value copula to infinite dimension. It is well defined, 
 yet  no direct construction
was known back then which lead the authors to class it among copula
models characterized by their motivation from multivariate
considerations.  
An application to Swiss rainfall
data in that paper bears witness of its versatility for extremal dependence
modeling.  
In the following, we show that  the extremal
$t$ process provides  a natural connection  
between two prominent max-stable model classes, namely Schlather's extremal Gaussian process (\cite{schlather2002models}) and
the Brown-Resnick process  defined in
\cite{brown1977extreme}  and revisited in a more general context in   \cite{kabluchko2009stationary}.
The connection to the
Brown-Resnick process was detailed for the multivariate context in
\cite{nikoloulopoulos2009extreme} and is related
to the study of elliptical triangular arrays with the H\"usler-Reiss
distribution (\cite{husler1989maxima}; \cite{falk2010laws})  as the
maximum attractor (\cite{hashorva2005elliptical}). It was then interpreted for the spatial context in \cite{davison2009modelling}.
The extremal $t$ dependence structure is further proposed for
semi-parametric inference in a multivariate context by
\cite{klueppelberg2007estimating} and \cite{klueppelberg2008semi}.
We conceive
a spectral representation for the extremal $t$ process that generalizes the one of the extremal
Gaussian process. It  renders direct simulation possible
for moderately large general degrees of freedom.

The remainder of the paper is organized as follows:
Section 2 gives some background in extreme value theory and
reviews results for elliptical distributions. Spectral constructions of multivariate extremal $t$ distributions and extremal $t$ processes are presented in
Section 3, along with a statement on the domain of attraction for
processes with finite-dimensional elliptical distributions. We conclude with a
discussion and potential future developments in Section 4. 

The following notational conventions shall apply in the remainder of
the paper:
If not stated otherwise,  operations on
vectorial arguments like maxima or arithmetic operations must be interpreted
componentwise. Vectors are typeset in bold face, in particular the
vector constants $\mathbf{0}=(0,...,0)^T$ and
$\mathbf{1}=(1,...,1)^T$.   Rectangular bounded or unbounded sets are given
according to notations like
$[\mathbf{u},\mathbf{v}] = [u_1,v_1]\times \cdots \times
[u_d,v_d]$ or $(\mathbf{0},\boldsymbol{\infty})=(0,\infty)\times
\cdots \times (0,\infty)$. The complementary set of a set $B$ in
$\mathbb{R}^d$ is written $B^c$. The truncation operator $x^+=\max(x,0)$ maps negative
values to $0$. The indicator function of a set $B$ is denoted by $\chi_B(\cdot)$.
\section{Extreme value theory}
For a more detailed account of max-stability and extreme value theory
in general we refer the reader  to 
the textbooks of \cite{beirlant2004statistics} and
\cite{de2006extreme}. 
\subsection{Max-stability}
Let $\mathbf{Z}, \mathbf{Z_1},\mathbf{Z_2},...$
be a sequence of independent and identically distributed
(iid) random vectors in $\mathbb{R}^d$ ($d\geq 1$) with  nondegenerate
univariate marginal distributions.  
We say that $\mathbf{Z}$ follows a max-stable distribution $G$ if sequences of  normalizing vectors
$\mathbf{a_n}>\mathbf{0}$ and $\mathbf{b_n}$ ($n=1,2,...$) exist such
that 
 the equality in distribution
\begin{equation}
\label{eq:1}
 \max_{i=1,...,n} \mathbf{a_n}^{-1}(\mathbf{Z_i}-\mathbf{b_n})
 \overset{d}{=} \mathbf{Z} \sim G
\end{equation}
holds for the componentwise maximum. 
A full characterization of multivariate max-stable distributions leads to rather
technical expressions. For our purposes, it is convenient to focus 
on common
$\alpha$-Fr\'echet marginal distributions $G_j(z_j) = \Phi_\alpha(z_j)= 
\exp(-z_j^{-\alpha})\chi_{(0,\infty)}(z_j)$ ($j=1,...,d$) for some \emph{tail index} $\alpha>0$. 
Monotone and parametric marginal 
transformations allow reconstructing all admissible univariate max-stable marginal
scales  in \eqref{eq:1} from this particular marginal scale. More
precisely, the
class of univariate max-stable distributions is partitioned into
the class of $\alpha$-Fr\'{e}chet distribution under   strictly increasing linear
transformations and further the so-called Gumbel and Weibull classes. 


With $\alpha$-Fr\'echet marginal distributions, the \emph{standard
  exponent measure} $\mathbb{M}$ can be defined on
$[\mathbf{0},\boldsymbol{\infty}) \setminus \{\mathbf{0}\}$ by $\mathbb{M}((\mathbf{0},\mathbf{z}]^c) = -\log \mathbb{P}(\mathbf{Z^\alpha}\leq
\mathbf{z})$ with the convention $-\log 0 = \infty$ and characterizes the
dependence structure in $G$ on a standardized scale; it is uniquely defined by  the dependence
function
$M(\mathbf{z})=\mathbb{M}((\mathbf{0},\mathbf{z}]^c)$ which takes the
value $\infty$ whenever $\min_j z_j=0$ such that $\mathbf{z}\not\in (\mathbf{0},\boldsymbol{\infty})$.
The
extremal coefficient $M(\mathbf{1})\in [1,d]$ can serve as  an indicator  of the
strength of extremal
dependence, ranging from full dependence 
associated with the value
$1$  to independence associated with the value $d$ (cf. \cite{schlather2003dependence}). 

In the infinite-dimensional domain, we call a  stochastic process $\mathbf{Z}=\{Z(s), s\in S \subset
\mathbb{R}^p\}$ ($p\geq 1$) with a non-empty Borel set $S$ max-stable if its finite-dimensional distributions are
max-stable. If $\mathbf{Z_1}, \mathbf{Z_2},...$ are iid copies of $\mathbf{Z}$, then sequences of functions $a_n(s)>0$ and $b_n(s)$
($n\geq 1$)  exist  such that
$\{\max_{i=1,...,n}a_n(s)^{-1}(Z_i(s)- b_n(s))\}
\overset{d}{=} \{ Z(s)\}$. 
\subsection{Domain of attraction}
Let  $\mathbf{X}, \mathbf{X_1}, \mathbf{X_2}, ...$ be a
sequence of  iid
random vectors  in $\mathbb{R}^d$ with distribution function $F$.  For suitably chosen normalizing sequences,  relation \eqref{eq:1} can hold
 asymptotically in the sense of distributional convergence with nondegenerate
marginal distributions in the
limit $\mathbf{Z}$: 
 \begin{equation}
\label{eq:MDA}
\max_{i=1,...,n} \mathbf{a_n}^{-1}(\mathbf{X_i}-\mathbf{b_n})
 \overset{d}{\rightarrow} \mathbf{Z} \quad (n\rightarrow \infty)\ . 
\end{equation}
We say that 
the distribution $F$  of $\mathbf{X}$ is  in
 the max-domain of attraction (MDA) of the max-stable distribution
 $G$ of $\mathbf{Z}$, or simply that $\mathbf{X}$ is in the MDA
 of $\mathbf{Z}$.
Normalizing sequences are not unique and the limit distribution $G$ is
unique up to  a linear transformation. If normalizing constants
can be chosen such that all the univariate marginal
distributions $G_j$  are of the same $\alpha$-Fr\'echet type, then the
particular choice of $\mathbf{b_n}=\mathbf{0}$ is admissible. In this
case,  the convergence in distribution \eqref{eq:MDA} is equivalent
to 
\begin{equation}
n\, \mathbb{P}(\mathbf{a_n}^{-1}\mathbf{X} \not\leq \mathbf{z}) \rightarrow
M(\mathbf{z}^{\alpha}) \quad \text{ for all } \mathbf{z} \in
(\mathbf{0},\boldsymbol{\infty})\ . 
\end{equation}
For $d=1$, we have $n\, \mathbb{P}(a_n^{-1}X \geq z)
\rightarrow z^{-\alpha}$ ($z>0$), and then $X$ is said to be \emph{regularly
  varying at $\infty$ with index $\alpha>0$} or just \emph{regularly
  varying} in the remainder of this paper, denoted as $X \in
\mathrm{RV}_\alpha$.  The normalizing sequence
can be chosen as $a_n = \inf\{x : P(X\geq x) \leq n^{-1}\}$.


For stochastic processes, the notion of MDA  is defined
in the sense of the convergence of all finite-dimensional
distributions according to \eqref{eq:MDA}. 
\subsection{A spectral representation for max-stable processes}
The commonly used models for  max-stable processes are generated with  so-called spectral
constructions whose first appearance dates back to the seminal paper of \cite{de1984spectral}.
\cite{schlather2002models} proposes to use a Poisson
process $\{V_i\} \sim \mathrm{PRM}(v^{-2}dv)$ on $(0,\infty)$  and
iid replicates $\mathbf{Q_i}$ of  an
integrable random process $\mathbf{Q}$,  independent of $\{V_i\}$ and
with $\mathbb{E} Q^+(s)=1$ ($s\in S$),  in order to construct the max-stable process
\begin{equation}
\label{eq:3}
\mathbf{Z} =\{Z(s)\} = 
\left\{\max_{i=1,2,...}V_i Q_i(s)\right\} \quad (s \in S)
\end{equation}
with univariate marginal distributions of type $\Phi_1$. 
It is possible to replace $Q_i(s)$ by the zero-truncated value
$Q_i^+(s)$ in this construction. Subsequently, without loss of
generality we assume that the points
$V_i$ are in
descending order such that $V_1\geq V_2\geq ...$ and $V_1\sim \Phi_1$.
We obtain extremal
Gaussian processes by choosing a centered and appropriately scaled
Gaussian process $\mathbf{W}$
for $\mathbf{Q}$. Other choices of $\mathbf{Q}$ were considered (cf. \cite{davison2009modelling}), leading for instance to
so-called Brown-Resnick processes (\cite{brown1977extreme}; \cite{kabluchko2009stationary}). 
The dependence function of $\mathbf{Z}$ for a finite number of points $s_1,...,s_d \in
S$ is
\begin{equation}
\label{eq:Mprocess}
M_{s_1,...,s_d}(\mathbf{z}) = \mathbb{E} \max_{j=1,...,d}
\left(z_j^{-1}Q^+(s_j)\right) \ . 
\end{equation} 

%
\subsection{Multivariate $t$ distributions and extremal dependence}
\subsubsection{Elliptical distributions}
\begin{definition}[Elliptically distributed random vectors]
\label{def:ell}
A random vector $\mathbf{X}$ in $\mathbb{R}^d$ is said to follow a
(non-singular)  elliptical
  distribution if it allows for a stochastic representation
$\mathbf{X} \overset{d}{=} \boldsymbol{\mu} + R_dA\mathbf{U}$ 
 with a deterministic location vector
$\boldsymbol{\mu}$, an invertible $d\times d$
 matrix $A$ that defines the dispersion matrix  $\Sigma=AA^T=(\sigma_{j_1j_2})_{1\leq j_1,j_2\leq d}$  and  a nondegenerate random variable $R_d\geq 0$
 independent from a random vector $\mathbf{U}$ uniformly distributed  on
 the Euclidean unit sphere 
$\{\mathbf{x} \in
\mathbb{R}^d \mid \mathbf{x}^T \mathbf{x}=1\}$. We call $R_d$ the
radial variable. 
\end{definition}
Here we prefer to remain within the framework of  quadratic and
non-singular $A$ to avoid an overly technical presentation for the
more general  cases (cf. \cite{anderson1990theory} for the general representation).
Random vectors following an elliptical multivariate $t$ distribution
are an important example: 
We say that an elliptically distributed random vector $\mathbf{X}$ in
$\mathbb{R}^d$ follows the
\emph{multivariate $t$ distribution with $\nu>0$ (general) degrees of
  freedom} if  $d^{-1}R_d^2\sim
F_{d,\nu}$, where $F_{d,\nu}$ is the $F$-distribution with degrees of
freedom $d$ and $\nu$.
We write $\mathbf{X}\sim t_\nu(\boldsymbol{\mu}, \Sigma)$ and
$\mathbb{P}(\mathbf{X}\leq \mathbf{x}) =  t_\nu(\mathbf{x} \mid \boldsymbol{\mu}, \Sigma)$.
The  multivariate $t$ distribution can be constructed as a variance
mixture of the multivariate normal distribution: 
With $\nu Y^{-1}\sim \mathrm{Gamma}(0.5\nu, 2)$ ($\nu>0$) and a multivariate normal random vector $\mathbf{W}\sim
N(\mathbf{0},\Sigma)$ that is independent of $Y$, we obtain
$\boldsymbol{\mu}+ \sqrt{Y}
\mathbf{W}\sim t_\nu(\boldsymbol{\mu},\Sigma)$, see for instance \cite{demarta2005t}. 
This construction is readily  generalized to the infinite-dimensional
setting on a domain $S$: If $\mathbf{W}$ is a centered Gaussian
process with domain $S$ and covariance
function $\mathrm{Cov}$ and  $\nu Y^{-1}\sim
\mathrm{Gamma}(0.5 \nu, 2)$ independent of $\mathbf{W}$, then we call the random process $
\sqrt{Y} \mathbf{W}$  a \emph{(centered) $t$ random process} on
$S$ which is
characterized by  the degree
  of freedom $\nu$ and the  dispersion function $\mathrm{Cov}$
  (cf. \cite{roislien2006t}). 

\subsubsection{The maximum attractor}
The multivariate $t$ distribution fulfills the MDA condition
\eqref{eq:MDA}. 
For normalizing constants $\mathbf{b_n}=\mathbf{0}$ and  $\mathbf{a_n}=n^{1/\nu}(\sigma_{jj}^{0.5}c_\nu^{1/\nu})_{j=1,...,d}$
with $c_\nu=\, \Gamma(0.5(\nu+1))^{-1}\nu^{-0.5\nu+1} \sqrt{\pi} \,
   \Gamma(0.5\nu)$ (cf. Table 2.1 on page 59 in \cite{beirlant2004statistics}), we obtain
 $\nu$-Fr\'echet marginal distributions in the max-stable limit
 distribution $G$ of the multivariate $t$ distribution $t_\nu(\boldsymbol{\mu}, \Sigma)$.
The dependence
 function of $G$ was derived by \cite{nikoloulopoulos2009extreme}: Denote by
 $\Sigma^*=(\sigma^*_{j_1,j_2})_{j_1,j_2}$ the correlation matrix  that corresponds to the
 dispersion matrix $\Sigma$, and by $\Sigma^*_{-j,-j} = (\sigma^*_{j_1,j_2})_{j_1\not=j,j_2\not=j}
 $ or $\Sigma^*_{j,-j}=(\Sigma^*_{-j,j})^T
 =(\sigma^*_{j_1,j_2})_{j_1=j,j_2\not=j}$ submatrices obtained by
 removing some of the rows or columns. Similarly for vectors, we write
 $\mathbf{z}_{-j} = (z_1,...,z_{j-1},z_{j+1},...,z_d)^T$.
Then 
\begin{equation}
\label{eq:Mt}
M_{\nu,\Sigma^*}(\mathbf{z}) = \sum_{j=1}^d z_j^{-1}\,
t_{\nu+1}\left((\mathbf{z}_{-j}/z_j)^{\nu^{-1}} \mid
  \Sigma^*_{-j,j},
  (\nu+1)^{-1}\left(\Sigma^*_{-j,-j}-\Sigma^*_{-j,j}\Sigma_{-j,j}^{*T}\right)\right) \ . 
\end{equation}
We refer to the max-stable limit as
\emph{extremal $t$ distribution}, and we call \emph{extremal $t$
  process} the max-stable limit of a $t$ random process which is a
generalization of the extremal $t$ distribution. Its dependence
structure is characterized by the set of dependence functions for all
finite-dimensional distributions, hence by a general degree of freedom
$\nu>0$ and the correlation function 
\[
\mathrm{Cov}^*(s_{j_1},s_{j_2})=\left[\mathrm{Cov}(s_{j_1},s_{j_1})
  \mathrm{Cov}(s_{j_2},s_{j_2})\right]^{-0.5}\mathrm{Cov}(s_{j_1},s_{j_2})
\quad (s_{j_1},s_{j_2}\in S)
\]
which corresponds to the dispersion function $\mathrm{Cov}$ and
determines the matrices $\Sigma^*$ for all finite-dimensional distributions.
Similar to the copula approach in the
multivariate domain, it is convenient to  call extremal $t$ process any
max-stable process whose set of dependence functions is the
same. 

In general, a regularly varying  radial variable $R_d \in \mathrm{RV}_\alpha$ ensures
the MDA condition for an elliptical distribution and  is a
necessary and sufficient condition for the presence of asymptotic
dependence, see Theorem 4.3 in \cite{hult2002multivariate}.  
Moreover, Theorem 3.1 of \cite{hashorva2006regular} establishes the equivalence of regular
variation of $R_d$ to regular variation  of any  component $X_j$ of $\mathbf{X}$.  
The index of regular variation $\alpha$ is the same across components and the
radial variable,  and the dependence
function $M$ depends only on $\alpha$ and the correlation matrix
$\Sigma^*$. 
For elliptical distributions in the Gumbel or Weibull  MDA, 
one always obtains independence in
the max-stable limit $G$. 
Since the multivariate $t$
distribution covers the full range of indices $\alpha>0$ (equal to the
general degree of freedom $\nu$)  and
correlation matrices $\Sigma^*$, the extremal $t$ dependence
structure is exhaustive within the class of asymptotically dependent
elliptical distributions.

\section{Main results}
The following theorem establishes the extremal $t$ process as the
maximum attractor for processes with finite-dimensional elliptical
distributions and asymptotic dependence. 
 \begin{theorem}[Elliptical domain of attraction]
 Let $\mathbf{X}=\{X(s), s\in S\}$ be a random process that has
finite-dimensional elliptical distributions according to the
dispersion function $\mathrm{Cov}$.
Suppose that $|\mathrm{Cov}^*(s_{j_1},s_{j_2})| <1$ for all
$s_{j_1}\not=s_{j_2}$, where $\mathrm{Cov}^*$ is the correlation
function that corresponds to $\mathrm{Cov}$.  Assume that one of
the two following conditions is fulfilled for $\mathbf{X}$:
\begin{itemize}
  \item At least one of the finite-dimensional distributions 
    for $d\geq 2$
    is in a
    multivariate MDA with asymptotic dependence. 
    \item At least one of the univariate marginal distributions of $\mathbf{X}$ is regularly
      varying. 
\end{itemize}
Then  a max-stable limit process $\mathbf{Z}$
exists   for $\mathbf{X}$  and is an extremal $t$ process. Its dependence functions for finite-dimensional
    distributions are given by \eqref{eq:Mt}.
\end{theorem}
\begin{proof}
We play on the equivalence of the regular variation condition for the
radial variable or for one of the components.
It is clear from Theorem 4.2 in \cite{hult2002multivariate}
(in the following referred to as HL) and
the equivalence that our first condition entails the
second one. Now assume $X(s_0)$ is regularly varying. For any
bivariate vector $(X(s_0), X(s))^T$ with $s\not=s_0$, the equivalence
dictates that
the radial variable $R_2$ associated to the bivariate random vector is
regularly varying, hence there is asymptotic dependence due to
HL. Since $X(s)$ is also regularly varying, we can iterate this
argument for all bivariate 
vectors $(X(s_1),X(s_2))$ with $s_1\not=s_2$
to prove that there is bivariate asymptotic dependence. Applying anew  HL
for any collection of distinct sites $s_1,...,s_d$ yields regular
variation for the associated radial variable $R_d$. 
Consequently, all finite-dimensional distributions of $\mathbf{X}$ possess max-stable limit
distributions of the extremal $t$ type which then constitute the extremal $t$
limit process $\mathbf{Z}$.
\end{proof}

We now provide a multivariate spectral construction for the extremal $t$
distribution based on elliptical distributions. 
\begin{theorem}[Multivariate spectral construction]
\label{theorem:ell}
Suppose the following items are given:
\begin{itemize}
  \item a tail index $\alpha>0$ and a correlation matrix $\Sigma^*$, 
    \item iid replications $\mathbf{X_i}$ of an elliptically distributed random vector
      $\mathbf{X}=(X_1,...,X_d)$ with  dispersion matrix $\Sigma=\Sigma^*$  and location vector $\boldsymbol{\mu}=\mathbf{0}$ such that the expectation $m_\alpha=\mathbb{E}
      [(X_1^+)^\alpha]$ is non-null and finite and 
      \item a Poisson process $\{ V_i\}\sim
        \mathrm{PRM}(\alpha v^{-(\alpha+1)}dv)$ on
        $(0,\infty)$. 
\end{itemize}
Define the componentwise maximum
\begin{equation}
\label{eq:spectral}
   \mathbf{Z} = m_\alpha^{-\alpha^{-1}}\max_{i=1,2,...} V_i
   \mathbf{X_i} \ .
\end{equation}
Then $\mathbf{Z}$ follows the extremal $t$ distribution with
$\alpha$-Fr\'echet marginal distributions and dependence function $M_{\alpha,\Sigma^*}$.
\end{theorem}
\begin{proof}
Due to infinite number of points  in the Poisson process $\{V_i\}$, we can
replace $\mathbf{X}$ by $\mathbf{X}^+$ in the construction \eqref{eq:spectral}.
By taking $\mathbf{Z}$ to the power of $\alpha$, i.e. $m_\alpha^{-1}\max_{i=1,2,...} V_i^\alpha
   (\mathbf{X_i^+})^\alpha$,   we obtain a special
case of the construction \eqref{eq:3} which proves the max-stability
and the $\alpha$-Fr\'echet marginal distributions of $\mathbf{Z}$.
Since $V_1\sim\Phi_{\alpha}$, the radial variable in the elliptical
random vector $V_1\mathbf{X_1}$ is regularly varying with index $\alpha$
due to the variant of Breiman's theorem from Lemma 2.3 in
\cite{davis2008extreme}; thus  $V_1\mathbf{X_1}$ is in the MDA of the extremal $t$
distribution. At the same time,  $V_1\mathbf{X_1}$ is in the MDA of 
$\mathbf{Z}$  according to Lemma 3.1 in \cite{segers2012maxstable}. We conclude that
$\mathbf{Z}$ follows the extremal $t$ distribution with dependence
function \eqref{eq:Mt}. 
\end{proof}
As a direct application of Theorem \ref{theorem:ell}, we are now able to present one possible spectral representation of
extremal $t$ processes via the corresponding Gaussian process. 
\begin{corollary}[Spectral representation of extremal $t$ processes]
\label{coroll:spectral}
Suppose the following items are  given:
\begin{itemize}
  \item a tail index $\alpha>0$ and a correlation function
    $\mathrm{Cov}^*$, 
  \item  iid replications $\mathbf{W_i}$ of a standard Gaussian random
    field $\mathbf{W}$ on $S\subset {\mathbb{R}}^p$ with dispersion function
    $\mathrm{Cov}=\mathrm{Cov}^*$ and
    \item a Poisson process $\{ V_i\}\sim
        \mathrm{PRM}(\alpha v^{-(\alpha+1)}dv)$ on
        $(0,\infty)$.
\end{itemize}
Then the process defined by 
\begin{equation}
\label{eq:constr-t}
  \mathbf{Z}=\{Z(s)\} = \left\{{m_\alpha}^{-\alpha^{-1}}\max_{i=1,2,...}
  V_i W_i(s)\right\} \quad (s \in S)\ , 
\end{equation}
with $m_\alpha =
\sqrt{\pi}^{-1}2^{0.5(\alpha-2)}\Gamma(0.5(\alpha+1))$
 and $\Gamma(\cdot)$  the Gamma function, 
is an extremal $t$ process  with $\alpha$-Fr\'echet marginal
distributions. Its dependence structure is characterized by  $\alpha$ general degrees of freedom and
the correlation function $\mathrm{Cov}^*$. 
\end{corollary}
\begin{proof}
It remains to verify the value of $m_\alpha$. Using the
variable transformation $y=0.5x^2$ yields
\[
m_\alpha = \int_0^\infty x^{\alpha}(2\pi)^{-0.5} \exp(-0.5x^2)dx =
\int_0^\infty (2y)^{0.5\alpha}(2\pi)^{-0.5} \exp(-y)  (2y)^{-0.5}dy \ , 
\]
and gathering the involved constants leads to the desired representation of $m_\alpha$.
\end{proof}
\section{Discussion}
\label{sec:disc}
The $\alpha$-power
$\{Z^\alpha(s)\}$ of \eqref{eq:constr-t}  establishes unit Fr\'echet marginal
distributions as in the construction \eqref{eq:3} with
$\mathbf{Q}=(\mathbf{W}^+)^\alpha$. Clearly, we identify the extremal
Gaussian process ( \cite{schlather2002models}) for $\alpha=1$.
When $d=2$ and $\sigma^*_{12}=0$ in Theorem \ref{theorem:ell}, the range of the extremal coefficient
covers the open interval $(1.5, 2)$: 
For $\alpha=1$ corresponding to the extremal Gaussian process, the value is known to be $1+ 0.5\sqrt{2}$.
As the degree of freedom $\nu=\alpha$ tends to infinity in
\eqref{eq:Mt}, the univariate $t$-distribution converges towards the
normal distribution and its variance tends to $0$ such that
$M(\mathbf{1})=2\lim_{\nu\rightarrow\infty}  t_{\nu+1}(1 \mid 0, (1+\nu)^{-1})=2$ in \eqref{eq:Mt}.
As $\nu$ tends to $0$, we observe $M(\mathbf{1})=
2\lim_{\nu\rightarrow 0} t_{\nu+1}(1 \mid 0, (1+\nu)^{-1})
= 2t_1(1\mid 0, 1) = 2(\pi^{-1}\mathrm{arctan}(1) + 0.5) = 1.5$.
This helps understand the long-range dependence structure in models for
extremal $t$ processes since the applied correlation functions are
usually non-negative and approach $0$ as the distance between two points
increases to infinity.
 In particular, extremal $t$ processes can be considered more
flexible than extremal Gaussian processes or Brown-Resnick
processes which are both special cases; see \cite{davison2009modelling} for the case of the  Brown-Resnick
process which arises for some $\alpha$-dependent correlation structures as $\alpha$ tends
to infinity. 
Moreover, the formulation of the dependence function \eqref{eq:Mt} for the extremal Gaussian process
($\nu=\alpha=1$) is more general than the bivariate expressions obtained by
\cite{schlather2002models} and lends itself
more easily to interpretation.

Theorem \ref{theorem:ell} and Corollary \ref{coroll:spectral} allow us
to
simulate extremal $t$ distributions and processes with the
method devised in Theorem 4 of 
\cite{schlather2002models}, thus completing the range of  max-stable
models available for direct simulation. When the degree of freedom is large, the computational complexity may become very restrictive, making it difficult to assure a
good quality of simulation. However,  in this case the H\"usler-Reiss distribution
(\cite{husler1989maxima}; \cite{falk2010laws}) in the
multivariate case and the Brown-Resnick process in the
infinite-dimensional case could be adequate proxies
for some correlation structures. 
Future research should explore in more detail up to which 
degree of freedom $\alpha$
the simulation procedure is numerically feasible, and if the Brown-Resnick process
provides an adequate substitute around and beyond the "critical" value
of $\alpha$. The spectral construction \eqref{eq:constr-t} further opens the way
for tackling conditional simulation in the theoretical framework
developed by \cite{dombry2011regular} and applied in
\cite{dombry2011conditional} and \cite{dombry2012conditional}. 
Test procedures based on an estimate of   $\alpha$ could be devised to  check for  the nested submodels  of
the extremal Gaussian or  Brown-Resnick type.
\bibliographystyle{model2-names}
\bibliography{../SpectralMeasure/ref}
\end{document}